\documentclass{article}[11 pts]
\usepackage[english]{babel}
\usepackage{amsmath}
\usepackage{amssymb}
\usepackage{amsthm}
\usepackage{mathrsfs}
\usepackage[dvips]{graphicx}
\usepackage{amssymb,amsmath,eepic,epic,colordvi,mathrsfs}
\usepackage{float}
\usepackage{fullpage}
\usepackage{color}

\newcommand{\ignore}[1]{}

\theoremstyle {plain}

\newtheorem{theorem}{Theorem}

\newtheorem {corollary}[theorem]{Corollary}
\newtheorem{definition}[theorem]{Definition}
\newtheorem{lemma}[theorem]{Lemma}

\newcommand{\ita}[1]{\textit{#1}}

\newcommand{\into}{\longrightarrow}
\newcommand{\ifff}{\Longleftrightarrow}

\newcommand{\la}{\langle}
\newcommand{\ra}{\rangle}
\newcommand{\andd}{\wedge}

\newcommand{\f}{\mathbb}
\newcommand{\rss}{\la\psi_i:i\in \lambda\ra}

\newcommand{\newts}{{\sf newts }}

\newcommand{\update}{{\sf update }}
\newcommand{\scan}{{\sf scan }}
\newcommand{\fscan}{{\sf Fscan }}

\newcommand{\classify}{{\sf classify }}
\newcommand{\newflag}{{\sf newflag }}
\newcommand{\findmax}{{\sf find\_max }}
\newcommand{\scate}{{\sf scate }}

\begin{document}

\title{The $F$-snapshot Problem.}
\author{Gal Amram \\ Department of Computer
Science,\\ Ben-Gurion University, Beer-Sheva, Israel}\date{}

\maketitle

\begin{abstract}

Aguilera, Gafni and Lamport introduced the signaling problem in \cite{AGL1}. In this problem, two processes numbered $0$ and $1$ can call two procedures: \update and {\sf Fscan}. A parameter of the problem is a two-variable function $F(x_0,x_1)$. Each process $p_i$ can assign values to variable $x_i$ by calling {\sf update($v$)} with some data value $v$, and compute the value: $F(x_0,x_1)$ by executing an \fscan procedure. The problem is interesting when the domain of $F$ is infinite and the range of $F$ is finite. In this case, some ``access restrictions'' are imposed that limit the size of the registers that the \fscan procedure can access.  

Aguilera et al. provided a non-blocking solution and asked whether a wait-free solution exists. A positive answer can be found in \cite{sig}. The natural generalization of the two-process signaling problem to an arbitrary number of processes turns out to yield an interesting generalization of the fundamental snapshot problem, which we call the $F$-snapshot problem. In this problem $n$ processes can write values to an $n$-segment array (each process to its own segment), and can read and obtain the value of an $n$-variable function $F$ on the array of segments. In case that the range of $F$ is finite, it is required that only bounded registers are accessed when the processes apply the function $F$ to the array, although the data values written to the segments may be taken from an infinite set. We provide here an affirmative answer to the question of Aguilera et al. for an arbitrary number of processes. Our solution employs only single-writer atomic registers, and its time complexity is $O(n\log n)$, which is also the time complexity of the fastest snapshot algorithm that uses only single-writer registers.

\end{abstract}

\section{Introduction}

In this paper we introduce a solution to the $F$-snapshot problem, which is a generalization of the well-studied snapshot problem (introduced independently by Afek et al. \cite{snap1,snap2}, by Anderson \cite{comp} and by Aspnes and Herlihy \cite{PRAM}). A snapshot object involves $n$ asynchronous processes that share an array of $n$ segments. Each process $p_i$ can write values to the $i$-th segment by invoking an \update procedure with a value taken from some range of values: $\mathit{Vals}$, and can scan the entire array by invoking an instantaneous \scan procedure. For any function $F:\mathit{Vals}^n\to D$ (where $D$ is any set and $\mathit{Vals}^n$ is the set of $n$-tuples of members of $\mathit{Vals}$) the $F$-snapshot variant differs from the snapshot problem in that the \fscan operation has to return the value $F(v_0,\ldots,v_{n-1})$ of the instantaneous segment values $v_0,\dots,v_{n-1}$. That is in comparison to the standard \scan operation, which returns the vector of values that the segments store at an instantaneous moment. The $F$-snapshot problem is interesting only if we impose an additional requirement, without which it can be trivially implemented by applying the function $F$ (assumed to be computable) to the values returned by the standard \scan operation. This additional requirement, for the case $n=2$, was suggested by Aguilera, Gafni and Lamport \cite{AGL1} (see also \cite{AGL2}) in what they called there the signaling problem.
Thus, our $F$-snapshot problem is a generalization of both the standard snapshot problem and the signaling problem (generalizing this problem from the $n=2$ case to the general case of arbitrary $n$).

In the signaling problem the set $\mathit{Vals}$ can be assumed to be infinite, and the set $D$ (the range of $F$) is finite (and small). The requirement is that an \fscan operation uses only bounded registers. That is, registers that can store only finitely many different values (the \update operations may access unbounded registers). The signaling problem was formulated just for two processes in \cite{AGL1}, and a wait-free solution for this problem was left there as an open problem. Thus, solving the general $F$-snapshot sets quite a challenge. A wait-free solution to the signaling problem is given in \cite{sig}, and here we present a (wait free) solution to the general $F$-snapshot problem.

As the domain of $F$ may be infinite, an \update operation cannot access only bounded registers. Furthermore, the $F$-snapshot problem generalizes the signaling problem which provides a solution to the mailbox problem. Abraham and Amram \cite{uri-gal} showed that even the mailbox problem cannot be solved while only bounded registers are employed.

In \cite{AGL1}, the signaling problem is justified for efficiency reasons. We consider a case in which the processes write values to their segments taken from an infinite range, but they are interested in some restricted data regarding these values (for example, which process invoked the largest value, how many different values there are etc.). An $F$-snapshot implementation may be more efficient in these cases than a snapshot implementation, since it is not necessary to scan the entire array for extracting the required information, and it suffices to read only bounded registers. Efficiency is mostly guaranteed when the \fscan operations are likely to be invoked much more frequently than the \update procedures. 

Now we describe the $F$-snapshot problem formally. Let $P=\{p_0,\dots,p_{n-1}\}$ be a set of $n$-asynchronous processes that communicate through shared registers and let
$$F:\mathit{Vals}^n\into D$$  
be an $n$-variables computable function from a (possibly infinite) domain $\mathit{Vals}$, into $D=Rng(F)$. The problem is to implement two procedures:
\begin{enumerate}
\item $\update(v)$ - invoked with an element $v\in \mathit{Vals}$. This procedure writes $v$ to the $i$-th segment of an $n$-array $A$, when invoked by $p_i$.
\item $\fscan$ - returns a value $d\in D$. This procedure returns $F(A[0],\dots,A[n-1])$, in contrast to a $\scan$ procedure which returns the entire array: $(A[0],\dots,A[n-1])$.
\end{enumerate}

The implementation needs to satisfy the following requirements:
\begin{enumerate}
\item All procedures are wait free. That is, each procedure eventually returns, if the executing process keep taking steps.
\item If $D$ is finite, then only bounded registers are accessed during \fscan operations.
\end{enumerate}
The $F$-snapshot problem can be studied under various communication restrictions. As an example, the $f$-array implementation by Jayanti \cite{JAY} solves the $F$-snapshot problem as well, when the LL/SC primitive is employed. However, the LL/SC operation cannot be implemented from read/write operations \cite{HER}. Here we assume the simplest means of communication:

\begin{enumerate}
\item[3.] Only single-writer multi-reader atomic registers are applied.
\end{enumerate} 

For correctness of $F$-snapshot implementations, we adapt the well known Linearizability condition, formulated by Herlihy and Wing \cite{Lin}. Roughly speaking, an $F$-snapshot algorithm is correct if for any of its executions the following hold: Each procedure execution can be identified with a unique moment during its actual execution (named the linearization point), such that the resulting sequential execution belongs to a set of correct sequential executions: the sequential specification of the object. The sequential specification of the $F$-snapshot object includes all executions of the following atomic implementation, presented by a code for process $p_i$. The code uses an array $A[0..n-1]$.

\begin{center}

\begin{tabular}{l l}
\begin{minipage}[h]{60mm} 
\vspace{6mm}
{\large{$\update(v)$}}
\begin{tabbing}
***\=**\=**\=**\=**\=**\=*\=*\=*\=\kill
1.\> $A[i]:=v$\\

\end{tabbing}
\vspace{1mm}
\end{minipage}

&
\begin{minipage}[h]{60mm}
\vspace{6mm}
{\large $\fscan()$} 
\begin{tabbing}
***\=**\=**\=**\=**\=**\=*\=*\=*\=\kill
1.\> return $F(A[0],\dots,A[n-1])$\\
\end{tabbing}
\vspace{1mm}
\end{minipage}\\ 

\end{tabular}
\end{center}

One of the basic methods for proving linearizability is to identify each procedure execution with the execution of one of its actions, and to prove that these linearization points satisfy the requirements. However, in some cases the linearization points are not fixed, and may even be identified with actions executed concurrently by other processes (the queue implementation in \cite{Lin} forms an example of such an algorithm). Hence, this approach is not complete. Therefore, we present the linearizability condition in an equivalent way to the one described above, a way that fits better the correctness proof we provide here for our $F$-snapshot algorithm.

In an execution of an $F$-snapshot algorithm, the procedure executions are partially ordered by the precedence relation $<$. That is, if $A$ and $B$ are procedure executions, $A<B$ means that $A$ ends before $B$ begins. An execution is linearizable if the relation $<$ can be extended to a linear ordering $\prec$ that satisfies the sequential specification, presented in Figure \ref{sequential-specification}. An $F$-snapshot algorithm is correct if all its executions are linearizable.

\begin{figure}[H] 
\label{sequential-specification}
\begin{center}
\fbox{
\begin{minipage}[t]{120mm}

\begin{enumerate}

\item All procedure executions are partitioned into $\update$ and $\fscan$ operations. An $\update$ operation is invoked with a value $v\in \mathit{Vals}$ and an $\fscan$ returns a value $d\in D$.

\item Each procedure execution belongs to a unique process $p_i$, $i<n$.

\item Let $S$ be an \fscan operation and for each $i<n$ assume that $U_i$ is the last $p_i$-\update event that precedes $S$ in $\prec$. Assume that each $U_i$ is invoked with a value $v_i$. Then, $S$ returns $F(v_0,\dots,v_{n-1})$. 

\end{enumerate}

\end{minipage}

}
\end{center}
\caption{$F$-snapshot sequential specification}

\end{figure} 

In this paper, we present a solution to the $F$-snapshot problem. Each operation in our algorithm consists of $O(n\log n)$ actions addressed to the shared registers. Thus, the time complexity of our algorithm equals the snapshot implementation by Attiya and Rachman \cite{nlogn} which is, as far as we know, the fastest published snapshot implementation with single-writer registers.

\section{The $F$-snapshot Algorithm}
\label{algo-section}

First we explain the main ideas behind the algorithm. The reader may want to consider our explanations, while examining the code of the algorithm given in Figure \ref{Algo}, and its local procedures in Figure \ref{procedures}.

The crucial obstacle for solving the problem is that an \fscan procedure cannot access unbounded registers, but it is required to apply the function $F$ on values that are stored in unbounded registers. Thus, the computation of $F$ needs to be done during an execution of an \update operation. When process $p_i$ performs an \update operation invoked with a data value $val$, it writes $val$ into a snapshot object $V$ (line 2), scans this snapshot object (line 3), applies the function $F$ on the view it obtained and stores the outcome in a local variable $ans$ (line 4). Then, before it returns it writes the outcome it obtained into a snapshot object named $Flags$ (line 16). A scanner scans the snapshot object $Flags$ and it needs to choose the most up-to-date value among the values suggested by the processes. We need to provide the $Flags$ object with an additional information, so a scanner could decide correctly which value to return. However, this additional information needs to be taken from a finite range due to the problem limitations.

As a first attempt, one may suggest to use bounded concurrent timestamps \cite{ts} to label update events (see \cite{ts1},\cite{ts2},\cite{ts3}). The intuitive approach is to choose a timestamp after updating $V$ and before scanning $V$, or to use a \scate operation as in \cite{nlogn} and to choose a timestamp immediately after the \scate operation returns. The problem is that a scanner will return a value relaying on the order between the labeling operations and not on the order between the scan events addressed to the snapshot object $V$. Hence, this approach in its simplest form, will not succeed. 

Our goal is to provide the $Flags$ snapshot object with some additional bounded information, so that a scanner will be able to determine the right order between scans addressed to $V$ that precede the writes into the snapshot object $Flags$. Several synchronization mechanisms are used for achieving this goal.

\subsection{The Classify Mechanism}

For determining the ordering between \scan events addressed to $V$, we adapt common technique of counting \update events \cite{nlogn},\cite{LA},\cite{ISS}. When a process performs an \update operation, it increases a counter (line 1) and writes this counter to $V$ together with the data value with which the \update operation was invoked (line 2). After a process scans $V$, it sums these counters to obtain a natural number that reflects how recent its view is (line 9). This approach resembles the snapshot algorithm presented by Israeli, Shaham and Shirazi \cite{ISS}. They used this technique to implement a snapshot algorithm in which the time complexity of the \scan procedure is $O(n)$. In their construction, while executing a \scan operation, the executing process returns the view of the process that presents the latest activity, reflected by the largest sum.

Since an \fscan operation cannot access unbounded registers, we cannot adopt the discussed approach as it was used in \cite{ISS}. In our algorithm, reading these natural numbers is done while executing an \update procedure. The process writes the sum it obtained into a snapshot object named $ViewSum$ (lines 10,12), and scans $ViewSum$ to compare its view with the views obtained by the other processes. Afterwards, it classifies all other processes into two categories: $winners$ - the processes that posses a later view, reflected by a largest sum, and $losers$ - processes with less up-to-date view. This is done by calling the local \classify procedure, when processes id's are used for breaking symmetry. Eventually, these sets of $winners$ and $losers$ will be stored at the segment $Flags[i]$ (lines 15,16).

\subsection{The coloring Mechanism}

A scanner scans the $Flags$ array and it tries to find the most up-to-date view while considering the fields $Flags[i].winner$ and $Flags[i].losers$ for $i=0,\dots,n-1$. For any pair of processes $p_i$ and $p_j$, the scanner tries to understand which process's view is more recent. The problem is that the processes may provide contradicting information. As an example, a scanner may found that $j\in Flags[i].winners$ (which means that $p_i$ thinks that $p_j$'s view is more up-to-date then its view), but it is possible that also $i\in Flags[j].winners$. Namely, it is possible that both $p_i$ and $p_j$ think that the other process knows better.

The coloring mechanism ensures that the problem described above can occur only in some ``typical" executions (with which our next mechanism deals). The \update events by each process alternate between 3 possible colors: 0, 1 or 2 (line 1). Each process posses a three-fields variable, in correspondence to three colors, named $myview$. After a process sums the counters it sees (lines 3,9), it writes the sum it obtained into $myview[color]$ (line 10) and deletes data obtained in its second-previous \update operation (line 11) to erase a confusing information. Then, it writes the value that $myview$ stores into $ViewSum$ (line 12). Now, when process $p_i$ scans $ViewSum$, in each segment $ViewSum[j]$, it finds two integers. These are the sums that $p_j$ computed in its two previous \update events. When $p_i$ executes its local \classify procedure, it also writes the color it saw. For example, it writes $(j,c)$ to $winners$ for $c\in\{0,1,2\}$, if it read from $ViewSum[j][c]$ a number greater than the sum it obtained in line 9 (when id's are taken into account for breaking symmetry). Each process $p_i$ writes the values of its local sets $winners$ and $losers$ to $Flags[i]$, together with the color of the \update operation it is executing.

Coming back to our example, now each process also specifies the color of the \update operation it saw. If a scanner finds that $(j,c)\in Flags[i].winners$ it understands that $p_i$ saw in $ViewSum[j][c]$ an integer larger than the number it obtained. However, if the scanner sees that $Flags[j].color\neq c$, it just disregards $p_i$'s information.

\subsection{Adding Bounded Timestamps}

The coloring mechanism does not prevent entirely the possibility that processes will provide contradicting information. Assume as an example that a scanner finds that $Flags[i].color=c_i$, $Flags[j].color=c_j$, $(j,c_j)\in Flags[i].winners$ and $(i,c_i)\in Flags[j].winners$. Thus, both $p_i$ and $p_j$ claim that the other process is more up-to-date. When such a situation occurs, one of the processes provides reliable information. This is the process that scanned $ViewSum$ later before updating $Flags$.

When such a situation occurs, the processes use timestamps to inform which process is trustworthy. We use a simple timestamps system in which the timestamps are vertices of a a nine-vertices directed graph $G=(V_G,E_G)$. A detailed explanation can be found in chapter 2 of \cite{HS} or in \cite{ts}. The graph $G$ consists of three cycles, each cycle includes three vertices. In addition, there is an edge from each vertex at the $i$-th cycle to each vertex at the $i-1\pmod 3$ cycle. Formally, $V_G=\{(i,j): i,j\in\{0,1,2\}\}$, and there is an edge from $v=(i_1,j_1)$ to $u=(i_2,j_2)$ if $i_1=i_2$ and $j_1=j_2+1\pmod 3$, or $i_1=i_2+1\pmod 3$. The vertices of $G$ are named timestamps, and if $(v,u)\in E$ we say that $v$ dominates $u$ and we write $u<_{ts}v$. Intuitively, $v$ dominates $u$ means that the timestamp $v$ represents a later moment than the timestamp $u$. 

We can see that there no cycles of length two in $G$. In addition, for any two timestamps $v,u$, we can find a timestamp $w$ that dominates both $v$ and $u$. We take a function $next:V_G\times V_G\into V_G$ that satisfies this property. That is, for any timestamps $v,u$: $v<_{ts}next(v,u)$ and $u<_{ts}next(v,u)$.

Any process $p_i$ holds $n$ pairs of timestamps. Each pair consists of a new timestamp and an old timestamp. These pairs are stored in a snapshot object named $VTS$. When $p_i$ executes an \update operation it scans $VTS$ (line 5). Then, against each process $p_j$ it chooses a timestamp that dominates the pair of timestamps it read from $VTS[j][i]$, using the function $next$. $p_i$ stores the timestamp it obtained as its new timestamp, keeps its former timestamp available as its old timestamp and updates $VTS$ (consider lines 5-8 and the local procedure $\newts$). Finally, $p_i$ stores its $n$-vector of pairs of timestamps in $Flags[i]$ while updating the $Flags$ object (line 16).

Now, consider again the situation in which a scanner finds that two processes $p_i$ and $p_j$, provide contradicting information as described earlier. In this case, the scanner checks the timestamps that the processes present. The process that its new timestamp dominates the other process's new timestamp is the reliable one. More precisely, the scanner considers the timestamps $Flags[i].vts[j].new$ and $Flags[j].vts[i].new$. The information provided by the process with the later timestamp is the right information. These timestamps are used only when processes provide contradicting information. In other cases the timestamps do not necessarily reflect the right ordering between the processes' views.

\subsection{The Code}

Now we specify the code of the algorithm and we start by presenting the data structures and the type of the registers and variables. First, the algorithm use 4 snapshot objects.

\begin{enumerate}
\item $V$ - each entry $V[i]$ stores a pair: $(n,val)\in \f N\times \mathit{Vals}$. $val$ is the value with which the \update procedure is invoked, and $n\in \f N$ counts the number of \update operations invoked by $p_i$. Note that as these values are taken from an infinite range, an \fscan operation cannot access this object. Initially each segment stores the value $(0,x_0)$ for some fixed $x_0\in \mathit{Vals}$. 

\item $VTS$ - each entry $VTS[i]$ stores an $n$-array: $vts_i[0..n-1]$ of pairs of timestamps. Thus, each entry $vts_i[j]=(v,u)$ where $v,u$ are timestamps. The first field is denoted $vts_i[j].old$, while the second field is $vts_i[j].new$ i.e. $(v,u)=(vts_i[j].old,vts_i[j].new)$. Initially, each field $VTS[i][j]$ stores $(v_0,v_0)$ for some fixed $v_0\in V_G$.

\item $ViewSum$ - each entry $ViewSum[i]$ is a triple: $viewsum_i[0..2]$ of natural numbers when each entry may also store $null$. That is, $ViewSum[i]\in (\f N\cup\{null\})\times (\f N\cup\{null\})\times (\f N\cup\{null\})$. There are three fields in correspondence to three possible colors of the \update operations. The initial value of each segment $ViewSum[i]$ is $(0,null,null)$.

\item $Flags$ - this is the bounded object scanned during \fscan operations. Each entry $Flags[i]$ stores an element of type $flag$. The $flag$ type consists of five fields:
\begin{enumerate}
\item $flag.color \in \{0,1,2\}$. Initially this field is $0$.

\item $flag.vts$ - an $n$-array of pairs of timestamps. Initially all pairs are $(v_0,v_0)$. Recall that $(v_0,v_0)$ is also the initial value of each $VTS[i][j]$. 

\item $flag.winners$, $flag.losers$ - sets of pairs of the form: $(i,c)\in\{0,\dots,n-1\}\times \{0,1,2\}$. At the $i$-th segment  $flag.winners$ and $flag.losers$ are initialized to $\{(j,0) : i<j\}$ and $\{(j,0): j<i\}$ respectively.

\item $flag.ans\in D$. The initial value of this field is $F(x_0,x_0,\dots,x_0)$. Recall that $x_0\in \mathit{Vals}$ is the initial value of each entry $V[i]$.

\end{enumerate}

\end{enumerate}

Each process use several local variables:

\begin{enumerate}

\item $color\in \{0,1,2\}$. Initially $color=0$.

\item $counter,viewsum \in \f N$. The initial value of these variables is $0$. 

\item $val\in \mathit{Vals}$.

\item $ans\in D$.

\item $myview\in \f (N\cup\{null\})\times \f (N\cup\{null\})\times \f (N\cup\{null\})$. Initially $myview=(0,null,null)$.

\item $winners,losers$ - sets that contain elements from the range $\{0,\dots,n-1\}\times\{0,1,2\}$.

\item $vts_i[0 .. n-1]$ - an $n$-array of pairs of timestamps.

\item $ts.old,ts.new$ - timestamps.

\item Other variables that are used for storing information while scanning the snapshot objects (lines 3,5,13). The type of each such a variable is in correspondence to the type of the objects that are scanned.
\end{enumerate}

\begin{figure}[h]

\begin{tabular}{|l|l|}\hline
\begin{minipage}[h]{60mm} 
\vspace{6mm}
{\large{$\update(val)$}}
\begin{tabbing}
***\=**\=**\=**\=**\=**\=*\=*\=*\=\kill
1.\> $counter:=counter+1$, $color:=counter\mod 3$,\\
2.\> $V.\update(counter,val)$\\
3.\> $(v_0,\dots,v_{n-1})=V.\scan$\\
4.\> $ans:=F(v_0.val,\dots,v_{n-1}.val)$\\
5.\> $(vts_0,\dots,vts_{n-1}):=VTS.\scan$\\
6.\> for $j=0$ to $n-1$ do\\
7.\>\> $vts_i[j]:=\newts(vts_j[i],vts_i[j])$\\
8.\> $VTS.\update(vts_i)$ \\
9.\> $viewsum:=v_0.counter+\dots+v_{n-1}.counter$\\
10.\> $myview[color]:=viewsum$\\
11.\> $myview[color+1\pmod 3]:=null$\\
12.\> $ViewSum.\update(myview)$\\
13.\> $(view_0,\dots,view_{n-1}):=ViewSum.\scan$\\
14.\> $\classify(view_0,\dots,view_{n-1})$ \\
15.\> $flag:=\newflag()$ \\
16.\> $Flags.\update(flag)$\\
\end{tabbing}
\vspace{1mm}
\end{minipage}

&
\begin{minipage}[h]{60mm}
\vspace{4mm}

{\large $\fscan()$} 
\begin{tabbing}
***\=**\=**\=**\=**\=**\=*\=*\=*\=\kill
1.\> $(flag_0,\dots,flag_{n-1}):=Flags.\scan$\\
2.\> $winner:= \findmax(flag_0,\dots,flag_{n-1})$\\
3.\> return $flag_{winner}.ans$\\
\end{tabbing}
\vspace{1mm}
\end{minipage}\\ \hline

\end{tabular}

\caption{code for $p_i$}
\label{Algo}
\end{figure}

The algorithm use 4 local procedures:

\begin{enumerate}
\item \classify - gets $n$ triples of natural numbers as arguments. Each triple represents the amount of knowledge that the corresponding process has obtained in its recent \update operations. As each \update operation has a color from the set $\{0,1,2\}$, each entry is a triple in correspondence to three possible colors. This procedure constructs the sets $flag.winners$ and $flag.losers$ based on the considerations explained above.

\item {\sf newflag.} Crates a new flag before updating the $Flags$ object.
 
\item \newts - gets two pairs of timestamps: $pair_1,pair_2$ and returns a pair of timestamps, $pair_3$ such that $pair_3.new$ dominates both fields of $pair_1$, and $pair_3.old=pair_2.new$ 

\item \findmax - gets $n$ $flags$ as arguments and returns an element from $\{0,\dots,n-1\}$. This procedure is invoked during an \fscan operation and the element that this procedure returns is the id of the most up-to-date process.

\end{enumerate}

The procedures $\classify,$ \newflag and  \newts are presented in Figure \ref{procedures}. The procedure \findmax is discussed in the next subsection.

\begin{figure}[h]

\begin{tabular}[h]{|l|l|l|}
\hline
\begin{minipage}[t]{60mm}
\vspace{4mm}

$\classify(view_0,\dots,view_{n-1})$
\begin{tabbing}
***\=**\=**\=**\=**\=**\=*\=*\=*\=\kill
1.\> $winners:=$ \\ 
\> $\{(j,c) :  view_j[c]>viewsum\}\cup$\\  
\> $\{(j,c) :  view_j[c]=viewsum\andd i<j\}$\\  
2.\> $losers:=$\\ 
\> $\{(j,c) :  view_j[c]<viewsum\}\cup$\\  
\> $\{(j,c) :  view_j[c]=viewsum\andd i>j\}$\\ 

\end{tabbing}
\end{minipage}
&
\begin{minipage}[t]{50mm}
\vspace{4mm}

$\newflag()$
\begin{tabbing}
***\=**\=**\=**\=**\=**\=*\=*\=*\=\kill
1.\> $flag.color:=color$\\
2.\> $flag.vts:=vts_i$\\
3.\> $flag.winners:=winners$\\
4.\> $flag.losers:=losers$\\
5.\> $flag.peers:=peers$\\
6.\> $flag.ans:=ans$\\
7.\> return $flag$\\
\end{tabbing}

\end{minipage}

&
\begin{minipage}[t]{40mm}
\vspace{4mm}

$\newts((u,v),(u',v'))$
\begin{tabbing}
***\=**\=**\=**\=**\=**\=*\=*\=*\=\kill
1. $ts.old:=v'$\\
2. $ts.new=(next(u,v))$\\
3.\> return $(ts.old,ts.new)$
\end{tabbing}  
\end{minipage}

\\ \hline
\end{tabular}

\caption{Local procedures }
\label{procedures}
\end{figure}

\subsection {The Procedure \findmax}
\label{findmax-procedure}

This procedure is invoked during an execution of an \fscan event $S$, and it returns the id of the most up-to-date process. Thus, the process that executes $S$ returns the value $flag_i.ans$ in case that \findmax returns $i$.

The \findmax procedure of an \fscan operation $S$ gets $n$ flags as arguments: $flags(S):=(flag_0,\dots,flag_{n-1})$. The procedure returns a maximal element in relation $<_S\subseteq \{0,\dots, n-1\}\times \{0,\dots,n-1\}$ that we define here. The relation $<_S$ is defined by reference to $flags(S)$ in definition \ref{<_S}.

\begin{definition}
\label{conflict}
Let $p_i$ and $p_j$ be two processes and write: $flag_i.color=c_i$ and $flag_j.color=c_j$. We say that $p_i$ and $p_j$ are in conflict, if one of the following occurs:
\begin{enumerate}
\item $(j,c_j)\in flag_i.winners$ and $(i,c_i)\in flag_j.winners$.
\item $(j,c_j)\in flag_i.losers$ and $(i,c_i)\in flag_j.losers$.

\end{enumerate} 
\end{definition}

Definition \ref{conflict} is important since, as we shall prove, for each two processes $p_i$, $p_j$ and an \fscan event $S$, the $flag$ of one of these processes determines correctly the ordering between $p_i$ and $p_j$. That is, if $p_i$ is the reliable process and if (for example) $(j,c_j)\in flag_i.winners$ and the color in $p_j$'s $flag$ is $c_j$, than $p_j$ is indeed more up-to-date than $p_i$ (more precisely, the $ans$ field of $p_j$'s flag is more up-to-date) as indicated by $p_i$'s flag. The problem is that we do not know which process provides correct information among any pair of processes. However, this problem does not arise when the processes are not in conflict. When processes provide contradicting information we use the processes' timestamps to find the trustworthy process.   

\begin{definition}
\label{<_S}
Let $p_i$ and $p_j$ be two processes and write: $flag_i.color=c_i$, $flag_j.color=c_j$. $i<_S j$ if one of the following occurs:
\begin{enumerate}

\item $p_i$ and $p_j$ are not in conflict and $(i,c_i)\in flag_j.losers$.
\item $p_i$ and $p_j$ are not in conflict and $(j,c_j)\in flag_i.winners$
\item $p_i$ and $p_j$ in conflict, $flag_i.vts[j].new<_{ts}flag_j.vts[i].new$ and $(i,c_i)\in flag_j.losers$.
\item $p_i$ and $p_j$ in conflict, $flag_j.vts[i].new<_{ts}flag_i.vts[j].new$ and $(j,c_j)\in flag_i.winners$.
\end{enumerate}
\end{definition}

An element $i\in\{0,\dots,n-1\}$ is maximal in $<_S$ if there is no $j\neq i$ such that $i<_Sj$. The procedure $\findmax(flag_0,\dots,flag_{n-1})$ (line 2) returns a maximal element in $<_S$ (we shall prove that such a maximal element exists in any \fscan event). This procedure \findmax accesses only local variables and we omit the technical but easy implementation of this procedure.

\section{Correctness}
\label{correctness}

Fixing an execution $\tau$, we need to show that the precedence relation defined over the high level events in $\tau$, $<$ can be extended into a linear ordering, $\prec$ that belongs to the sequential specification of the $F$-snapshot object. In section \ref{findmax-procedure}, a relation $<_S$ was defined. An \fscan event $S$ returns a value stored in $Flags[j].ans$ where $j$ is maximal in relation $<_S$. Thus, for proving correctness we also need to show that for any \fscan event in $\tau$, $S$, $<_S$ admits a maximal element.  

Since the algorithm is wait-free, we may assume that all operations in $\tau$ are complete. Indeed, if there are pending operations in $\tau$, we can let the processes take additional steps and complete the pending operation. This way, an execution that extends $\tau$ is obtained. A linearization of the resulting execution admits a linearization of $\tau$ as well.

As explained in the preliminaries section, at the beginning of $\tau$, each process performs an initialization and writes initial values to the registers and variables. These initial high level events precede all other actions in $\tau$ and are considered as $\update$ events. If $I_j$ is such an initial event by $p_j$, the value of this event, $val(I_j)$ is $x_0$. Recall that the initial value of the entry $Vals[j]$ is $(0,x_0)$. 

The execution $\tau$ is a sequence of atomic actions addressed to the shared memory. Thus, the procedure executions addressed to the snapshot objects (e.g line 2 of the \update procedure) are not atomic and represent a sequence of actions that a process executes. However, by using a linearizable implementation for the snapshot objects, we may assume for convenience that all the procedure executions addressed to the snapshot objects are atomic. This assumption simplifies our proof since we do not need to speak about the linearization points of these operations and the corresponding extension of $<$. For further discussion about using linearizable implementations see \cite{AMP} and \cite{Lin}.

Our algorithm employs several snapshot objects. Thus, for preventing confusion, we use the notation $A.\update$ and $A.\scan$ to denote invocations of $\update$ and $\scan$ procedures addressed to object $A$. Note that an $A.\update(x)$ invocation by $p_i$ writes $x$ to the $i$-th segment of $A$.

If $e$ is a read (write) event executed by some of the processes, we use $val(e)$ to denote the value that the process read (wrote) in $e$. Similarly, if $e$ is an $A.\update$ event addressed to a snapshot object $A$, $val(e)$ is the value that the executing process wrote to the corresponding segment of $A$, and if $e$ is an $A.\scan$ event, $val(e)$ is the vector of elements that $e$ returns. Any low level event $e$ belongs to a unique high level event, which is an \update or an \fscan event by some of the processes. We use $[e]$ to denote this event. It is clear that $e\in[e]$.

The following notations are important in our proof:
\begin{enumerate}

\item For an $A.$\scan event $e$ on a snapshot object $A$, we define $\mu_j(e)$ to be the maximal $A.\update$ event by $p_j$ that precedes $e$. Thus, $val(e)[j]=val(\mu_j(e))$ for any $A.$\scan event $e$. 

\item Let $U$ be an \update event and $p_i$ a process. If $U$ is an initial \update event we set $\alpha_i(U)=I_i$, the initial $p_i$-\update event. Otherwise, $\alpha_i(U)$ is the $p_i$-\update event in which $p_i$ wrote to $V[i]$ the value that was read from $V[i]$ in $U$. That is:
$$\alpha_i(U)=[\mu_i(V.\scan(U))]$$
where $V.\scan(U)$ is the (unique) $V.\scan$ event in $U$, which corresponds to the execution of line 3 in the code of the \update procedure. 

\item Let $S$ be an \fscan event in which $winner=j$ (the invocation of \findmax in $S$ returns $j$). Let $e$ be the (unique) $Flags.$\scan event in $S$, and let $U_j=[\mu_j(e)]$. For a process id $i$, we define $\alpha_i(S)=\alpha_i(U_j)$.

\item Let $S$ be an \fscan event. For a process id $i$, $\beta_i(S)$ is the $p_i$-\update event that wrote to $Flags[i]$ the value read in $S$. That is, $\beta_i(S)=[\mu_i(Flags.\scan(S))]$ where $Flags.\scan(S)$ is the (unique) $Flags.\scan$ event in $S$.
\end{enumerate}

Two easy observations that will be useful later are the following:

\begin{lemma}
\label{alpha-is-a-fix-point}
For each $p_i$-\update event $U$, $\alpha_i(U)=U$.
\end{lemma}

\begin{lemma}
\label{<-implies-<alpha}
Let $U_1$ and $U_2$ be two \update events such that $U_1< U_2$. Then, for each process id $i$, $\alpha_i(U_1)\leq\alpha_i(U_2)$. 
\end{lemma}

Lemma \ref{alpha-is-a-fix-point} holds since the $V.$\update operation in $U$ precedes the $V.$\scan operation (lines 2 and 3). Lemma \ref{<-implies-<alpha} holds since the $V.$\scan event in $U_1$ precedes the $V.$\scan event in $U_2$.

We fix a compete \fscan event $S$ and we shall prove that there is a maximal element in relation $<_S$. For each process id $i$, write $U_i=\beta_i(S)$, $flag_i=val(\mu_i(Flags.\scan(S)))$, and $c_i=flag_i.color$. That is, $U_i$ is the $p_i$-\update event that wrote to $Flags$ the value read in $S$, $flag_i$ is the value that $p_i$ wrote to $Flags[i]$ in $U_i$ and $c_i$ is the value of the color field of $flag_i$. According the initial values of the snapshot object $Flags$, and by the code of the \classify procedure, the following hold:

\begin{lemma}
For a pair $(j,c)\in \{0,\dots,n-1\} \times \{0,1,2\}$ and a process $p_i$, at most one of the following occurs:

\begin{enumerate}
\item $(j,c)\in flag_i.winners$.
\item $(j,c)\in flag_i.losers$.

\end{enumerate}
\end{lemma}

\begin{corollary}
\label{no-2-cycles-in-<C}
If $i<_S j$, then $\neg(j\ <_S i)$.
\end{corollary}
\begin{proof} 
Consider definition \ref{<_S}, and observe that relation $<_S$ between $i$ and $j$ is determined only by the $flag$ of one of these processes. Hence, this is a consequence from the previous lemma. 
\end{proof}

If $U_i$ is not the initial \update event $I_i$, when $p_i$ executed $U_i$, it computed a natural number while executing line 9. Let $m_i$ denote this number. If $U_i=I_i$, define $m_i=0$. We argue that $m_i$ reflects correctly how recent $p_i$'s view is.

\begin{lemma}
\label{counters-means-all}
For two processes $p_i$ and $p_j$, if $(m_i,i)<(m_j,j)$ at the lexicographic order, then for each process id $k$, $\alpha_k(U_i)\leq\alpha_k(U_j)$.
\end{lemma} 

\begin{proof}
If $U_i=I_i$, then for each process id $k$, $\alpha_k(U_i)$ is the first $p_k$-\update event and the lemma hold. If $U_j=I_j,$ then $m_j=0$ which implies that $m_i=0$. Thus, $U_i=I_i$ and we are done. It is left to deal with the case that $U_i\neq I_i$ and $U_j\neq I_j$.

Towards a contradiction, assume that $\alpha_k(U_j)< \alpha_k(U_i)$ for some process id $k$. We conclude that the $V.\scan$ operation in $U_i$ occurred after the $V.\scan$ operation in $U_j$. Therefore, the counter that $p_i$ read from each field $V[t].counter$ is larger than the counter that $p_j$ read
 (note that the $l$-th \update operation by each process writes $l$ to this field). Hence, $m_j\leq m_i$. However, since the integer that $p_i$ read from $V[k].counter$ is strictly larger than the one that $p_j$ read, $m_j<m_i$ in contradiction to the assumption that $(m_i,i)<(m_j,j)$.   
\end{proof}

During the execution of $S$, for each two processes $p_i$ and $p_j$, the process that executes $S$ decides whether $i<_S j$ or $j<_S i$. The decision is made upon the values of these processes' $flags$. The next lemmas show that for each such a pair of processes, at least one of these processes' $flags$ provides reliable information. That is to say, for some process (say, $p_i$) the following occurs:
\begin{itemize}
\item If $flag_i.ans$ is more up-to-date than $flag_j.ans$, then $(j,c_j)\in flag_i.losers$.
\item If $flag_i.ans$ is less up-to-date than $flag_j.ans$, then $(j,c_j)\in flag_i.winners$.
\end{itemize}  

\begin{lemma}
\label{later-reads-correctly}
Let $p_i$ and $p_j$ be two processes such that $U_i\neq I_i$. Let $e_i\in U_i$ be the \update of $ViewSum$ in $U_i$ (line 12) and let $e_j$ be the \update of $ViewSum$ in $U_j$. Let $e$ be the \scan event of $ViewSum$ in $U_i$ (line 13). If $e_j<e_i$, then one of the following holds:
\begin{enumerate}
\item $\mu_j(e)=e_j$ or,

\item $\mu_j(e)=e'>e_j$ and there is no $p_j$-\update event between $U_j=[e_j]$ and $[e']$.
\end{enumerate}
\end{lemma}

\begin{proof}
Since $e_j<e_i$ and since (by the code) $e_i<e$, we see that $e_j<e$ and hence, $e_j\leq \mu_j(e)$. Thus, we need to show that there is at most one $ViewSum.\update$ event by $p_j$ between $e_j$ and $e$. 

Assume for a contradiction that $e'$ and $e''$ are two $ViewSum.\update$ events by $p_j$ such that 
$$e_j<e'<e''<e.$$ 
Each $ViewSum.\update$ event belongs to a unique  \update event so there are two different $p_j$-\update operations $U'=[e']$ and $U''=[e'']$. Recall that $[e_j]=U_j$ and observe that:
$$\beta_j(S)=U_j<U'<e''<e.$$

Now, the $Flags.\scan$ event in $S$ occurs after $e$, so it reads the value written to $Flags[j]$ in $U'$ or in a later event. We have:
$$\beta_j(S)=U_j<U'\leq [\mu_j(Flags.\update(S))]$$
in contradiction to the definition of $\beta_j(S)$.   
\end{proof}

We conclude:

\begin{lemma}
\label{write-after-knows}
Let $p_i$ and $p_j$ be two processes such that $U_i\neq I_i$. Let $e_i\in U_i$ be the \update of $ViewSum$ in $U_i$, let $e_j$ be the \update of $ViewSum$ in $U_j$, and let $e$ be the \scan event of $Views$ in $U_i$. If $e_j<e_i$, then $p_i$ reads $m_j$ from $ViewSum[j][c_j]$ in $e$, and in addition:
\begin{enumerate}
\item If $(m_i,i)<(m_j,j)$ (at the lexicographic order), then $(j,c_j)\in flag_i.winners$.
\item If $(m_i,i)>(m_j,j)$ (at the lexicographic order), then $(j,c_j)\in flag_i.losers$.

\end{enumerate}
\end{lemma}

\begin{proof}
$e\in U_i$ is the $ViewSum.\scan$ event in $U_i$ by $p_i$. By the previous lemma, since $e_j<e_i$ there is at most one write to $ViewSum$ between $e_j$ and $e$. Thus, the value that $p_j$ wrote to $ViewSum[j][c_j]$ (which is $m_j$) has not been ``deleted" (consider lines 10-12 in the \update procedure). $p_i$ reads $m_j$ from $ViewSum[j][c_j]$ and the lemma follows from the code of the \classify procedure and from lemma \ref{counters-means-all}.
\end{proof}

So far, we have proved that for any two processes $p_i$ and $p_j$, one of these processes' $flags$ provides reliable information. Namely, the process that wrote later to $ViewSum$ during the \update events $U_i$ and $U_j$. The next lemma easily stems.
  
\begin{lemma}
\label{not-in-conflict}
Let $p_i$ and $p_j$ be two processes which are not in conflict (the definition is in section \ref{findmax-procedure}). If  $(m_i,i)<(m_j,j)$ lexicographically, then $i<_S j$. 
\end{lemma}

\begin{proof}
First assume that $m_j=0$. In this case, $U_j=I_j$ the initial \update event. In addition, since $(m_i,i)<(m_j,j)$, also $m_i=0$ thus $U_i=I_i$ as well. We conclude that $i<j$ and according to the initial values of the registers we get that $(j,0)\in Flags[i].winners$ and $(i,0)\in Flags[j].losers$. In addition, $c_i=c_j=0$ and hence, $i<_Sj$ as required.

Now assume that $m_j>0$ and conclude that $U_j\neq I_j$. Write $e_i\in U_i$ - the \update of $ViewSum$ in $U_i$ and respectively, $e_j$ is the \update of $ViewSum$ in $U_j$. Assume w.l.o.g. that $e_j<e_i$ and observe that $U_i$ is not the initial event either. By lemma \ref{write-after-knows}, $(j,c_j)\in flag_i.winners$. Since the processes are not in conflict the claim holds.
\end{proof}

Our next goal is to prove the same for the case that the processes are in conflict. If the processes are in conflict, we know by the previous lemmas that one of them provides reliable information. Recall that in this case, the definition of $<_S$ is according to the $flag$ of the process that presents a later timestamp. We need to show that the process with the later timestamp is the reliable one, namely the one that wrote later to $ViewSum$.

\begin{lemma}
Let $p_i$ and $p_j$ be two processes. Let $e_i\in U_i$ be the \update of $ViewSum$ in $U_i$, let $e_j$ be the \update of $ViewSum$ in $U_j$, and assume that $e_j<e_i$. If $p_i$ and $p_j$ are in conflict, then $flag_j.vts[i].new<_{ts} flag_i.vts[j].new$. 
\end{lemma}

\begin{proof}
First, note that $U_i\neq I_i$. Indeed, if $U_i=I_i$ we get that also $U_j=I_j$ (since $e_j<e_i$) which implies that the processes are not in conflict.

For the rest of the proof we assume that $U_j\neq I_j$. If $U_j=I_j$, then similar (and simpler) argument can be applied. Let $s_j$ be the \scan of $ViewSum$ in $U_j$. By lemma \ref{write-after-knows}, $p_i$ reads $m_j$ from $ViewSum[j][c_j]$ in $U_i$, but since $p_i$ and $p_j$ are in conflict, conclude that $p_j$ read some $k\neq m_i$ from $ViewSum[i][c_i]$ in $s_j$. Hence, 
\begin{equation}
\label{e-doesnt-read-e_i}
\mu_i(s_j)\neq e_i.
\end{equation}

Since $e_j<e_i$ and (by the code) $e_j<s_j$, either $e_j<s_j<e_i$ or $e_j<e_i<s_j$. We claim that the first option occurs and $s_j<e_i$. Assume otherwise and use equation \ref{e-doesnt-read-e_i} to conclude that $e_i<\mu_i(s_j)$. Note that there can be at most one $ViewSum.\update$ event by $p_i$ that follows $e_i$ and precedes $s_j$ (consider the arguments in the proof of lemma \ref{later-reads-correctly}), and hence $s_j$ reads from $ViewSum[i]$ the value of this event. However, by the code of the \update procedure, the \update operation by $p_i$ that follows $U_i$ also writes $m_i$ to $ViewSum[i][c_i]$. Thus, if $e_i<s_j$, then $p_j$ reads $m_i$ from $ViewSum[i][c_i]$ in $s_j$, and this is in contradiction to the assumption that the processes are in conflict.

Now we claim that there is a $p_i$-$ViewSum.\update$ event between $s_j$ and $e_i$. Indeed, assume not and let $e'$ be the last $ViewSum.\update$ event by $p_i$ that precedes $e_i$. By our assumption we have $\mu_i(s_j)=e'$.  $e'$ belongs to the last $p_i$-\update event that precedes $U_i$ and hence the color of this \update event is $c_i-1\pmod 3$. Therefore, $val(e')[c_i]=null$. We conclude that $p_j$ read $null$ from $ViewSum[i][c_i]$ in $s_j$ and this contradicts the fact that $p_i$ and $p_j$ are in conflict.

We see that there is a $ViewSum.\update$ event by $p_i$ between $s_j$ and $e_i$, say $e_i'$. That is,
$$s_j<e_i'<e_i.$$
Write $[e_i']=U'$, a $p_i$-\update event and note that $U'<U_i$. Therefore, $e_i'<U_i$ and hence
$$s_j<U_i.$$

Now, let $t_j\in U_j$ be the (unique) $VTS.\update$ event in $U_j$ (line 8 in the code) and write $val(t_j)[i]=(x,y)$. Observe that since $t_j\in U_j$, $(x,y)$ is also the value of $flag_j.vts[i]$. Let $s_i\in U_i$ be the (unique) $VTS.\scan$ event in $U_i$ (line 5) and note that since $e_i'<U_i$, $e_i'<s_i$. By the code and by our conclusions we have:
$t_j<s_j<e_i'<s_i$ thus $$\mu_j(s_i)\geq t_j.$$
Note that there is at most one $VTS.\update$ event by $p_j$ between $t_j$ and $s_i$ since otherwise, we would have $\beta_j(S)\neq U_j$. Furthermore, if there is such an event, it writes to $VTS[j][i]$: $(y,z)$ for some vertex $z\in V_G$ (consider the \newts code). Let $(a,b)$ denotes the value of $VTS[i][j]$ before the execution of $U_i$.

\begin{itemize}
\item[Case 1.] $\mu_j(s_i)=t_j$ and hence $p_i$ reads in $s_i$ from $VTS[j][i]$: $(x,y)$. Thus, $p_i$ writes in $U_i$ to $Flags[i].vts[j]:\newts((x,y),(a,b))=(b,next(x,y))$. Since $(next(x,y),y)\in E_G$, $flag_j.vts[i].new<_{ts} flag_i.vts[j].new$ as required.

\item[Case 2.] $\mu_j(s_i)>t_j$ and hence $p_i$ reads in $s_i$ from $VTS[j][i]$: $(y,z)$. In this case $p_i$ writes in $U_i$ to $Flags[i].vts[j]:\newts((y,z),(a,b))=(b,next(y,z))$. Since also $(next(y,z),y)\in E_G$, $flag_j.vts[i].new<_{ts} flag_i.vts[j].new$. We see that the lemma holds in this case as well.
\end{itemize}
\end{proof}

The previous lemma shows that if two processes are in conflict and their $flags$ provide contradicting information, the $flag.vts$ fields determined correctly which among the two processes is the reliable one. The conclusion is that relation $<_S$ determines correctly which process presents the most up-to-date view in its $flag.ans$ field.

\begin{lemma}
\label{alpha-implies-c}
Let $p_i$ and $p_j$ be two processes. Then, $(m_i,i)<(m_j,j) \ifff i<_S j$. 
\end{lemma} 

\begin{proof}
First assume that $(m_i,i)<(m_j,j)$ and we shall prove that $i<_S j$. If $p_i$ and $p_j$ are not in conflict, then this is the case of lemma \ref{not-in-conflict}. If $p_i$ and $p_j$ are in conflict, let $e_i$ be the \update of $ViewSum$ in $U_i$ and let $e_j$ be the \update of $ViewSum$ in $U_j$. Assume w.l.o.g. that $e_j<e_i$. By lemma \ref{write-after-knows}, $(j,c_j)\in flag_i.winners$. By the previous lemma $flag_j.vts[i].new<_{ts}flag_i.vts[j].new$ thus by definition, $i<_S j$.

Now, for the other direction, assume that $i<_Sj$. If $(m_j,j)<(m_i,i)$, then we get that also $j<_S i$ in contradiction to corollary \ref{no-2-cycles-in-<C}. Thus, $(m_i,i)<(m_j,j)$ as required.
\end{proof}

It is easy to see that any two \update events $U$ and $U'$, are comparable in $\leq_\alpha$. For verifying this observation, assume w.l.o.g. that the $V.\scan$ event in $U'$ occurs after the $V.\scan$ event in $U$. Clearly, $U\leq_\alpha U'$ in this case. Therefore, we conclude:.

\begin{corollary}
\label{main-corollary}
There is a maximal element in $<_S$ and hence, the \findmax procedure in $S$ returns some $j<n$.
\end{corollary}

\begin{proof}
Take (the unique) $j$ such that $(m_j,j)$ is maximal at the lexicographic order over $\{(m_0,0),(m_1,1),\dots,(m_{n-1},n-1)\}$. By the previous lemma, $j$ is maximal in relation $i<_S j$. 
\end{proof}

Now we are ready to show that $\tau$ is a linearizable. Wew define a linear ordering $\prec$ on the set of all high-level events in $\tau$. First, we define $\prec$ over the \update events. Then, we define $\prec$ between update and \fscan events and finally, we define $\prec$ over \fscan events.

\begin{enumerate}
\item For two \update operations $U,U'$, we set $U\prec U'$ if the write to $V$ in  $U$ precedes the write to $V$ in $U'$. That is, the executions of the $V.$\update operations are the linearization points of the \update events.

\item Let $S$ be an \fscan event. For each \update event $U$, we decide if $U\prec S$ or $S\prec U$ by choosing an \update event to linearize $S$ immediately after it. 

For each process $p_i$, write  $U_i=\alpha_i(S)$. We linearize $S$ immediately after the initial \update events $U_0,\dots,U_{n-1}$. More precisely, we linearize $S$ after the maximal element in $\prec$ over the set $\{U_0,\dots,U_{n-1}\}$.

\item It is left to define $\prec$ over the \fscan events. First, we consider pairs of \fscan events $S,S'$ such that $S$ was linearized after an \update event $U$ and $S'$ was linearized after an \update event $U'\neq U$. In this case, if $U\prec U'$, we set $S\prec S'$.

Now, for each \update event $U$, we take the \fscan events linearized after $U$, $S_1,\dots,S_m$, and we extend $\prec$ on these events in some arbitrary way that extends $<$ over $S_1,\dots,S_m$.
\end{enumerate}

It is easy to verify that $\prec$ is a linear ordering, now we verify that $\prec$ extends $<$. Consider two high-level events $A<B$. We shall prove that $A\prec B$. The claim is trivial when $A$ and $B$ are \update events as these events were linearized in correspondence to an execution of an atomic instruction. We need to deal with the cases that $A$ and $B$ are both \fscan events, or one of them is an \fscan event and the other is an \update event.
\begin{enumerate}

\item[Case 1.] $A=U$ an \update event, say by $p_i$ and $B=S$ an \fscan event. For each process $p_k$, let $U_k$ denote the $p_k$-\update event that wrote to $Flags[k]$ the value read in $S$. i.e. $U_k=\beta_k(S)$. Note that $U\leq U_i$. Assume that the procedure \findmax in $S$ returned $j$ thus $i\leq_S j$ and $\alpha_i(S)=\alpha_i(U_j)$. Use lemmas \ref{alpha-implies-c}, \ref{counters-means-all} and \ref{alpha-is-a-fix-point} to observe that:
$$U\leq U_i=\alpha_i(U_i)\leq \alpha_i(U_j)=\alpha_i(S).$$
 Recall that $S$ was linearized after $\alpha_i(S)$ and hence, since $\prec$ extends $<$ over \update events,
$$U\preceq U_i=\alpha_i(U_i)\preceq\alpha_i(U_j)\prec S$$
which implies that $U\prec S$. 

\item[Case 2.] $A=S$ an \fscan event and $B=U$ an \update event, say by $p_i$. Note that since $S<U$, $U\neq I_i$ the initial $p_i$-\update event. To show that $S$ was linearized before $U$, we need to show that for each process $p_k$, $\alpha_k(S)\prec U$. 

Assume that the procedure \findmax in $S$ returns $j$ and write $U_j=\beta_j(S)$. For a process $p_k$, write $U_k=\alpha_k(S)=\alpha_k(U_j)$. If $U_j=I_j$, then $U_k=I_k$ and then it is clear that $U_k\prec U$ as required. Otherwise, let $e_k$ be the write to $V$ in $U_k$, let $r$ be the $V.$\scan event in $U_j$ and let $e$ be the $V.$\update event in $U$. Obviously, $e_k<r$. Since $\beta_j(S)=U_j$, $\neg(S<r)$. But since $S<U$, we conclude that $r<e$. As a result, $e_k<e$ which implies that $U_k\prec U$ as required.

\item[Case 3.] $A=S$ and $B=S'$ are both \fscan event. For proving that $S$ is linearized before $S'$ we show that for each process $p_i$, $\alpha_i(S)\preceq\alpha_i(S')$. Assume that the procedure \findmax in $S$ returns $j$ and the procedure \findmax in $S'$ returns $k$. Write $U_j=\beta(S)$ and $U_k'=\beta_k(S')$. Hence, $\alpha_i(S)=\alpha_i(U_j)$ and $\alpha_i(S')=\alpha_i(U_k')$. Write $U_j'=\beta_j(S')$ and use lemmas \ref{alpha-implies-c} and \ref{counters-means-all} to conclude that $\alpha_i(U_j')\leq \alpha_i(U_k')$. Since $S<S'$, $U_j\leq U_j'$ thus by lemma \ref{<-implies-<alpha} we get, $$\alpha_i(S)=\alpha_i(U_j)\leq\alpha_i(U_j')\leq\alpha_i(U_k)=\alpha_i(S).$$ Hence $\alpha_i(S)\leq\alpha_i(S')$ and $\alpha_i(S)\prec\alpha_i(S')$ follows.

\end{enumerate}

It is left to prove that the properties of the sequential specification are satisfied. It is easy to see that each \fscan event $S$ returns $F(val(\alpha_0(S)),\dots,val(\alpha_{n-1}(S)))$. Therefore, we need to verify that for each process $p_i$, $\alpha_i(S)$ is the maximal $p_i$-\update event that precedes $S$ in $\prec$. Since $S$ was linearized after the events $\alpha_0(S),\dots,\alpha_{n-1}(S)$, clearly $\alpha_i(S)\prec S$ for each process $p_i$. 

Towards a contradiction, assume that for some process $p_i$, $U\neq\alpha_i(S)$ is the maximal $p_i$-\update event that precedes $S$ in $\prec$. Hence,
$$ \alpha_i(S)\prec U\prec S.$$ 
We conclude that there is a process $p_k$ such that
$$\alpha_i(S)\prec U\prec \alpha_k(S)$$ since otherwise, $S$ would have linearized before $U$.
Note that $\alpha_k(S)\neq I_k$. Assume that the procedure \findmax in $S$ returns $j$ and write $U_j=\beta_j(S)$. Thus, $\alpha_k(S)= \alpha_k(U_j)$. Since $\alpha_k(S)$ is not the initial $p_k$-event, also $U_j\neq I_j$. 

Write $\alpha_i(S)=\alpha_i(U_j)=U_i$ and $\alpha_k(S)=\alpha_k(U_j)=U_k$. Let $e_i$ be the $V.\update$ operation in $U_i$, let $e$ be the $V.\update$ operation in $U$ and let $e_k$ be the $V.\update$ operation in $U_k$. Since $U_i\prec U\prec U_k$, we have
$$e_i<e<e_k.$$ 
Now, let $r$ be the $V.\scan$ event in $U_j$. By definition, $\mu_k(r)\in U_k$ thus $\mu_k(r)=e_k$ and in particular 
$$e_k<r.$$
As a result, $e_i<e<r$ thus $\mu_i(r)\neq e_i$ in contradiction to $\alpha_i(U_j)=U_i$.

\section{Conclusions}
\label{conclusions}

The snapshot object is a special case of the $F$-snapshot object while choosing the parameter $F$ to be the identity function. The $F$-snapshot object also generalizes the signaling object \cite{AGL1} from the case that there only two processes to an arbitrary number of processes. We present here a wait free solution to this problem. Our algorithm uses several snapshot objects thus its complexity measures depend on the exact implementations of this objects.

When processes communicate through shared read/write registers, any \fscan implementation must include $\Omega(n)$ operations addressed to the shared memory \cite{Jay-lower-bound}. Furthermore, regarding the snapshot object, Israeli and Shirazi \cite{IS} proved the same lower bound for \update implementations. As the snapshot object is a special case of the $F$-snapshot object, this lower bound holds for the $F$-snapshot object as well. 

The exact implementations for the snapshot objects that are used in our algorithm, determine the time complexity of the algorithm. However, since the $Flags$ object is accessed during \fscan operations, it is required to use a snapshot implementation that employs only bounded registers, in case that only finitely many different values are invoked by the processes. As an example, the first algorithm in \cite{snap2} violates this requirement since it uses a field named $seq$ that grows infinitely, while the second algorithm in \cite{snap2} satisfies this property. 

For efficiency, we can use the implementation by Attiya and Rachman in \cite{nlogn}. In section 4.4 of \cite{nlogn}, the authors explain how to transform their algorithm into a snapshot implementation which satisfies the requirements discussed here. Namely, into a snapshot implementation that uses only bounded registers, in case that finitely many data values may be written to the segments of the object. Thus, our algorithm can be implemented with time complexity $O(n\log n)$ which is, as far as we know, the time complexity of the most efficient published snapshot algorithm that uses only single-writer registers.

It is known that the snapshot object can be implemented with time complexity $O(n)$ when multi-writers are allowed as Inoue, Masuzawa, Chen and Tokura proved \cite{Lin-snap}. Inoue et al. present an algorithm that solves the lattice agreement problem. Then, the reduction by Attiya, Herlihy and Rachman \cite{LA}, provides a linear snapshot implementation with multi-writer registers. The problem is that this reduction requires unbounded memory. Hence, the $F$-snapshot limitations forbid using this implementation for the $Flags$ snapshot object in our algorithm. Therefore, the question if there is a linear $F$-snapshot implementation using multi-writers is not answered here, although there is a linear snapshot implementation that uses multi-writer registers.

By the essence of the problem, a natural complexity measure for an $F$-snapshot implementation is the size of the flags - the bounded registers that are accessed during an \fscan operation. This ``flags complexity" depends on: $n$ - the number of processes and $|D|$ - the number of distinct values that $F$ may return. Since at least $\log |D|$ bits are required to represent $|D|$ different values, it is not difficult to prove that the flags complexity of any solution is  $\Omega(n\log{|D|})$. For proving this claim, consider an $n$-variable function $F$ that satisfies the following: If we assign values to $n-1$ variables, then any element from $D=Rng(F)$ can be obtained by some assignment to the last variable. For example, the function $f(a_1,\dots,a_n)=a_1+\dots+a_n \pmod D$ satisfies this requirement. Now, since each process can execute an \update operation and change the result of an ensuing \fscan event into any element from $D=Rng(F)$, the size of each flag is at least $\log |D|$ bits and the lower bound holds. However, we do not know to prove any non-trivial lower bound on the size of the flags.

For calculating the flags complexity of our algorithm, for convenience, we may assume that $D=\{0,\dots,|D|-1\}$. Otherwise, we can take a bijective function $f:D\into \{0,\dots,|D|-1\}$ and replace the function $F$ with $f\circ F$. The $flag$ type consists of several fields. The field $flag.ans$ contains elements from $D$ and hence it requires $\log |D|$ bits. The size of the other fields depends only on $n$ - the number of processes. The set fields: $flag.winners,flag.losers$ can be represented using $3n$ bits, when each bit corresponds to a pair $(i,c)$ of process id and a color. Therefore, all  other fields are of size $O(n)$ thus the values that the processes write to $Flags$ require $O(n+\log |D|)$ bits. However, as $Flags$ is a snapshot object, the implementation for this object uses additional fields and the largest one contains a view. Namely, it contains a vector of $n$ elements, each entry store a value of the type that the processes write to the snapshot object. Thus, the size of each $flag$ is actually $O(n^2+n\log|D|)$ bits. The total size of the flags - which is the flags complexity of the algorithm is $O(n^3+n^2\log|D|)$. We believe that this can be significantly improved.

In our algorithm, the \fscan procedure accesses only bounded registers due to the problem constrains and the \update procedure access unbounded registers (otherwise the problem is unsolvable).  The segments of the snapshot object $V$ store elements from $\mathit{Vals}$ (which might be infinite), and counters that grow infinitely. Hence, if we take a function $F$ with a finite domain, the \update procedure will still access unbounded registers. Thus, in those cases, it is better to use some other implementation such as the bounded version of the algorithm in \cite{nlogn}. An interesting question that arises is whether there is a $F$-snapshot algorithm that satisfies both properties:

\begin{enumerate}
\item If $F$ has a finite range, then the \fscan procedure accesses only bounded registers. 
\item If $F$ has a finite domain, then only bounded registers are accessed. 

\end{enumerate}

\end{document}